\documentclass[twocolumn, superscriptaddress, nofootinbib, floatfix]{revtex4-2}

\usepackage{lmodern}
\usepackage{float}
\usepackage[colorlinks, allcolors=blue]{hyperref}
\usepackage{amsthm}
\usepackage{graphicx}
\usepackage{tasks,amsthm,amsmath,amssymb,mathtools,parskip,enumerate,subcaption,multirow,bbold,physics}
\usepackage[title]{appendix}
\usepackage[nameinlink,noabbrev]{cleveref}

\newtheorem{theorem}{Theorem}
\newtheorem{lem}[theorem]{Lemma}
\newtheorem{definition}[theorem]{Definition}
\newtheorem{remark}[theorem]{Remark}
\newtheorem{prop}[theorem]{Proposition}
\newtheorem{conj}[theorem]{Conjecture}

\theoremstyle{definition}

\newtheorem*{eg*}{Example~\ref{eg:simplest_nontrivial} cont}

\theoremstyle{remark}

\begin{document}

\title{Maximal Non-Kochen-Specker Sets and \\ a Lower Bound on the Size of Kochen-Specker Sets}

\author{Tom Williams}
\email{tommywilliams765@gmail.com}
\affiliation{School of Informatics, Quantum Software Lab, University of Edinburgh, UK}

\author{Andrei Constantin}
\email{andrei.constantin@physics.ox.ac.uk}
\affiliation{Rudolf Peierls Centre for Theoretical Physics, University of Oxford, Parks Road, Oxford OX1 3PU, UK}

\begin{abstract}
\noindent

The challenge of determining bounds for the minimal number of vectors in a three-dimensional Kochen-Specker (KS) set has captivated the quantum foundations community for decades. This paper establishes a weak lower bound of 10 vectors, which does not surpass the current best-known bound of 24 vectors. By exploring the complementary concept of large non-KS sets and employing a probability argument independent of the graph structure of KS sets, we introduce a novel technique that could be applied in the future to derive tighter bounds. Additionally, we highlight an intriguing connection to a generalisation of the moving sofa problem in navigating a right-angled hallway on the surface of a two-dimensional sphere.

\end{abstract}

\maketitle

\section{Introduction}\label{Introduction}

The Kochen-Specker (KS) theorem~\cite{kochen_specker_1967} demonstrated that it is impossible to explain quantum theory by non-contextual hidden variable theories, where measurement outcomes are assumed to be independent of the specific set of compatible observables, or context, in which they are measured. As such, quantum theory is inherently contextual.
%

Since then, quantum contextuality has been a major topic of research, both as a fundamentally non-classical phenomenon and as a key resource for quantum computation~\cite{Howard2014, Raussendorf2013, Abramsky2017}. Several mathematical frameworks have been developed to study contextuality, including operational approaches~\cite{spekkens_2005}, sheaf theory~\cite{Abramsky2011}, the framework of contextuality-by-default~\cite{Dzhafarov2016_2}, graph and hypergraph theory~\cite{Cabello2014,Acin2015}. All proofs of the Kochen-Specker (KS) theorem, as well as the free will theorem of Conway and Kochen~\cite{conway2008strong} rely on the identification of small sets of measurements that exhibit contextuality, known as KS sets\footnote{Note, however, that state-independent contextuality can be demonstrated without the use of KS sets, as illustrated, for instance,  in~Refs.~\cite{yu2012state, bengtsson2012kochen}.}. While KS sets are best known for demonstrating state-independent contextuality, more modern applications include their necessity, under certain constraints, for quantum advantage in tasks like one-shot zero-error communication~\cite{cubitt2010improving}, for optimal quantum strategies in bipartite scenarios~\cite{cabello2025simplest}, as well as other applications~\cite{liu2024equivalence}.

The current record for the smallest three-dimensional KS set was found by Conway and Kochen and consists of 31 directions \cite{peres1997quantum}. 
Going beyond three dimensions, it has been shown that the KS set of minimal cardinality occurs in dimension four and has 18 directions~\cite{cabello1996bell, PhysRevLett.124.230401}.
%
The main contribution of this paper is to explore the opposite notion of large, measurable non-contextual sets (non-KS sets). We find that the maximal three-dimensional non-KS set must cover at least a fraction of 0.8978 of the area of the two-sphere. We use this information to derive a lower bound on the number of directions contained in any KS set in three dimensions. Although this lower bound does not improve upon the current lower bound of 24 directions~\cite{li2024sat}, nor upon the previous bound of 22 directions~\cite{uijlen2016kochen, arends2011searching}, our argument is qualitatively different and does not rely on exhaustively checking the existence of KS graphs.

\section{Background}

{\bfseries Hidden Variables.} Quantum mechanics is non-deterministic: rather than predicting explicit outcomes of measurements as it is the case in classical mechanics, it merely predicts probabilities for such outcomes, whose indeterminacy is constrained by the uncertainty principle. Suppose this indeterminacy is not fundamental but is instead due to our incomplete knowledge of the system (akin to classical statistical mechanics); in this case, we may postulate the existence of a \textbf{valuation map} ${v}_{\psi}(\widehat{A})$, which is to be understood as the value of the observable (self-adjoint operator) $\widehat{A}$ when the quantum state is $\psi \in \mathcal{H}$, where $\mathcal{H}$ is the Hilbert space of the system.

It is impossible to make progress beyond this postulate without imposing additional conditions, the most natural one being the following. For any function $f:\mathbb{R} \to \mathbb{R}$,
\begin{equation}
    {v}_{\psi}(f(\widehat{A})) = f({v}_{\psi}(\widehat{A}))~.
\end{equation}
From this condition, it follows (see,~e.g.,~Ref.~\cite{isham2001lectures}) that for any commuting self-adjoint operators $\widehat{A},\widehat{B}$ and any state $\psi \in \mathcal{H}$,
\begin{equation}
\begin{aligned}
    {v}_{\psi}(\widehat{A}+\widehat{B})&={v}_{\psi}(\widehat{A})+{v}_{\psi}(\widehat{B})~, \\
    {v}_{\psi}(\widehat{A}\widehat{B})&={v}_{\psi}(\widehat{A}){v}_{\psi}(\widehat{B})~, \\
    {v}_{\psi}(\widehat{\mathbb{1}})&=1~.
\end{aligned}
\end{equation}

Consider now an arbitrary orthonormal basis (ONB) of~$\mathcal H$, $\{\ket{{e}_{1}},\ket{{e}_{2}},...,\ket{{e}_{d}}\}$ where $d$ is the dimension of the Hilbert space. The projectors ${\widehat{P}}_{i}:=\ket{{e}_{i}}\bra{{e}_{i}}$ satisfy 
\begin{equation}
    {\widehat{P}}_{i}^{2}={\widehat{P}}_{i}~, \qquad \sum {\widehat{P}}_{i}=\mathbb{1}~.
\end{equation}
Hence ${v}_{\psi}({\widehat{P}}_{i})={v}_{\psi}({\widehat{P}}_{i}^{2})={{v}_{\psi}({\widehat{P}}_{i})}^{2}$ and ${v}_{\psi}({\widehat{P}}_{i}) = 0 \text{ or } 1$.
Consequently, for any  ONB $\{\ket{{e}_{1}},\ket{{e}_{2}},...,\ket{{e}_{d}} \}$ exactly one projection operator ${\widehat{P}}_{i}=\ket{{e}_{i}}\bra{{e}_{i}}$ must have value~1 and all others must have value~0. This is equivalent to the colouring problem in which no two orthogonal directions can both be coloured~1 and exactly one direction in any orthogonal basis is coloured~1, all other directions being coloured~0.

\vspace{4pt}
{\bfseries The Kochen-Specker (KS) Theorem.} The KS theorem states that no such valuation map exists for Hilbert spaces of dimension greater than 2. 

Proofs of the KS theorem often take the form of constructing Kochen-Specker sets, defined as finite sets of rank-one projectors for which it is impossible to consistently assign outcomes such that
\begin{enumerate}[(KS.1)]
  \item no two orthogonal directions have outcome 1;
  \item exactly one direction in any orthogonal frame has outcome 1.
\end{enumerate}

Kochen and Specker's original proof~\cite{kochen_specker_1967} consisted of 117 directions in three dimensions built from several smaller sets with 8 directions each and known in the literature as Clifton graphs. \autoref{fig:1} shows a Clifton graph, denoted by ${\mathcal{G}}_{1}$, together with its bundle diagram. Note that some edges in the bundle diagram are triangles rather than lines as some measurement contexts consist of three measurements.
Clifton graphs are the smallest example of so-called \textbf{01-gadgets} which are sets of directions such that in any $\{0,1\}$-colouring there exist two non-orthogonal directions that cannot both be assigned the colour 1 (see Ref.~\cite{ramanathan2020gadget} for details). Clifton graphs enjoy the following useful property.

\begin{figure}[ht!]
\centering
  \includegraphics[width=.48\linewidth]{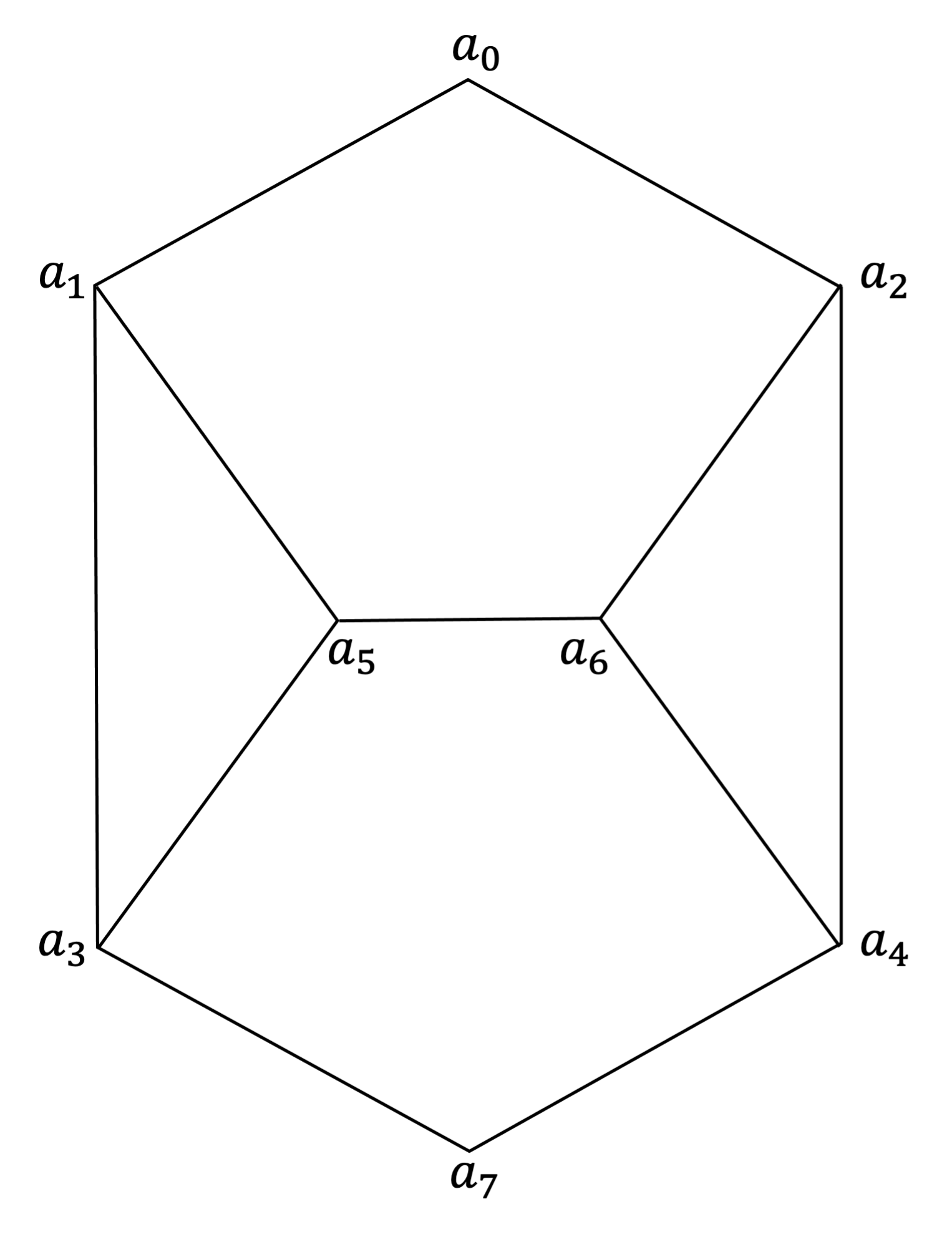}
  \includegraphics[width=.48\linewidth]{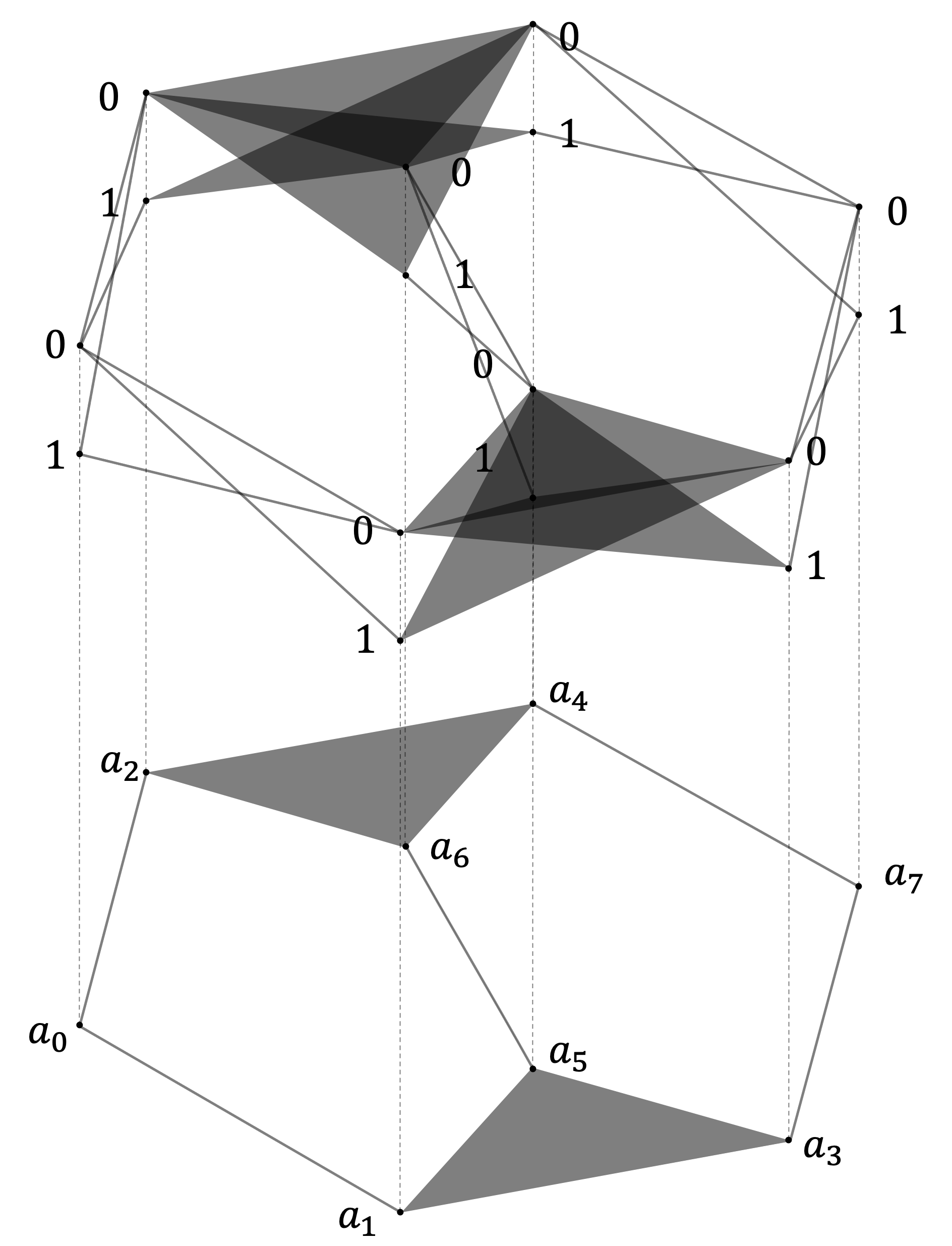}
  \caption{Left: the Clifton graph ${\mathcal{G}}_{1}$ where vertices are directions and edges join orthogonal directions. Right: the bundle diagram of ${\mathcal{G}}_{1}$.}
\label{fig:1}
\end{figure}

\begin{prop}
Directions ${a}_{0}$ and ${a}_{7}$ in \autoref{fig:1} cannot both have outcome 1.
\end{prop}
\begin{proof}
    Suppose $v({a}_{0}) = v({a}_{7}) = 1$, then condition (KS.1) implies $v({a}_{1}) =  v({a}_{2}) = v({a}_{3}) = v({a}_{4}) = 0$, then condition (KS.2) implies $v({a}_{5}) = v({a}_{6}) = 1$. However, condition (KS.1) does not allow this, hence the contradiction.
\end{proof}

Clifton graphs suffer from the limitation that the angle between the directions ${a}_{0}$ and ${a}_{7}$ in \autoref{fig:1} cannot exceed $\text{arccos}(\frac{1}{3})$ (see, e.g.~, Ref.~\cite{ramanathan2020gadget} for a proof). As such, several Clifton graphs are required to produce a KS set (or, using the language of Ref.~\cite{Abramsky2011}, a strongly contextual set). Kochen and Specker's proof uses 15 Clifton graphs and is shown in \autoref{fig:2}.

\begin{figure}[h!]
    \centering
    \includegraphics[width=.7\linewidth]{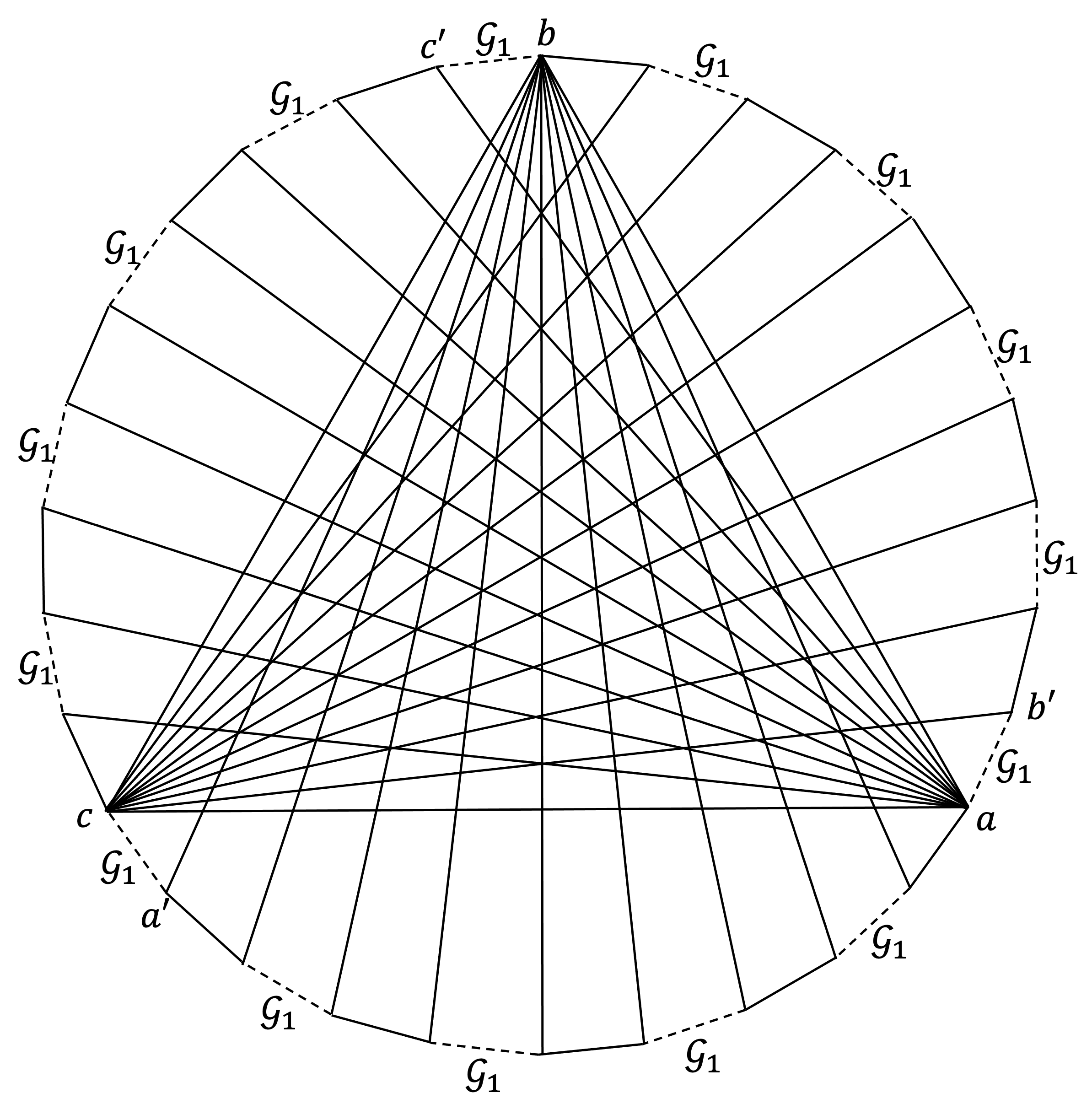}
    \caption{Kochen and Specker's 117 direction proof. Each solid line connects two orthogonal directions and each dashed line represents a Clifton graph ${\mathcal{G}}_{1}$. The following identifications hold: $a=a'$, $b=b'$, and $c=c'$.}
    \label{fig:2}
\end{figure}

\vspace{4pt}
{\bfseries A New KS Set.} A conceptually simple KS set consisting of 168 directions can be built using the idea of nested Clifton graphs~\cite{ramanathan2020gadget}, defined as follows. 

Consider the Clifton graph ${\mathcal{G}}_{1}$ but replace the edge $({a}_{5},{a}_{6})$ with a second Clifton graph ${\mathcal{G}}_{1}'$ with ${a}_{0}'={a}_{5}$ and ${a}_{7}'={a}_{6}$. The condition that directions ${a}_{0}$ and ${a}_{7}$ cannot both have outcome 1 is maintained. This process can be repeated iteratively. Denote an $n$-times nested Clifton graph by ${\mathcal{G}}_{n}$ with vectors ${a}_{0}^{(1)},...,{a}_{7}^{(1)},...,{a}_{0}^{(n)},...,{a}_{7}^{(n)}$, where ${a}_{5}^{(k)}={a}_{0}^{(k-1)}$ and ${a}_{6}^{(k)}={a}_{7}^{(k-1)}$ (see \autoref{fig:3a}). Then the following result holds (see Ref.~\cite{ramanathan2020gadget} for a proof).

\begin{prop}\label{Proposition 2.2}
In a nested Clifton graph ${\mathcal{G}}_{n}$ the directions ${a}_{0}^{(n)}$ and ${a}_{7}^{(n)}$ cannot both have outcome 1. The maximum angle between ${a}_{0}^{(n)}$ and ${a}_{7}^{(n)}$ is $\text{arccos}(\frac{n}{n+2})$.
\end{prop}

\begin{figure}[ht]
  \centering
  \includegraphics[width=.7\linewidth]{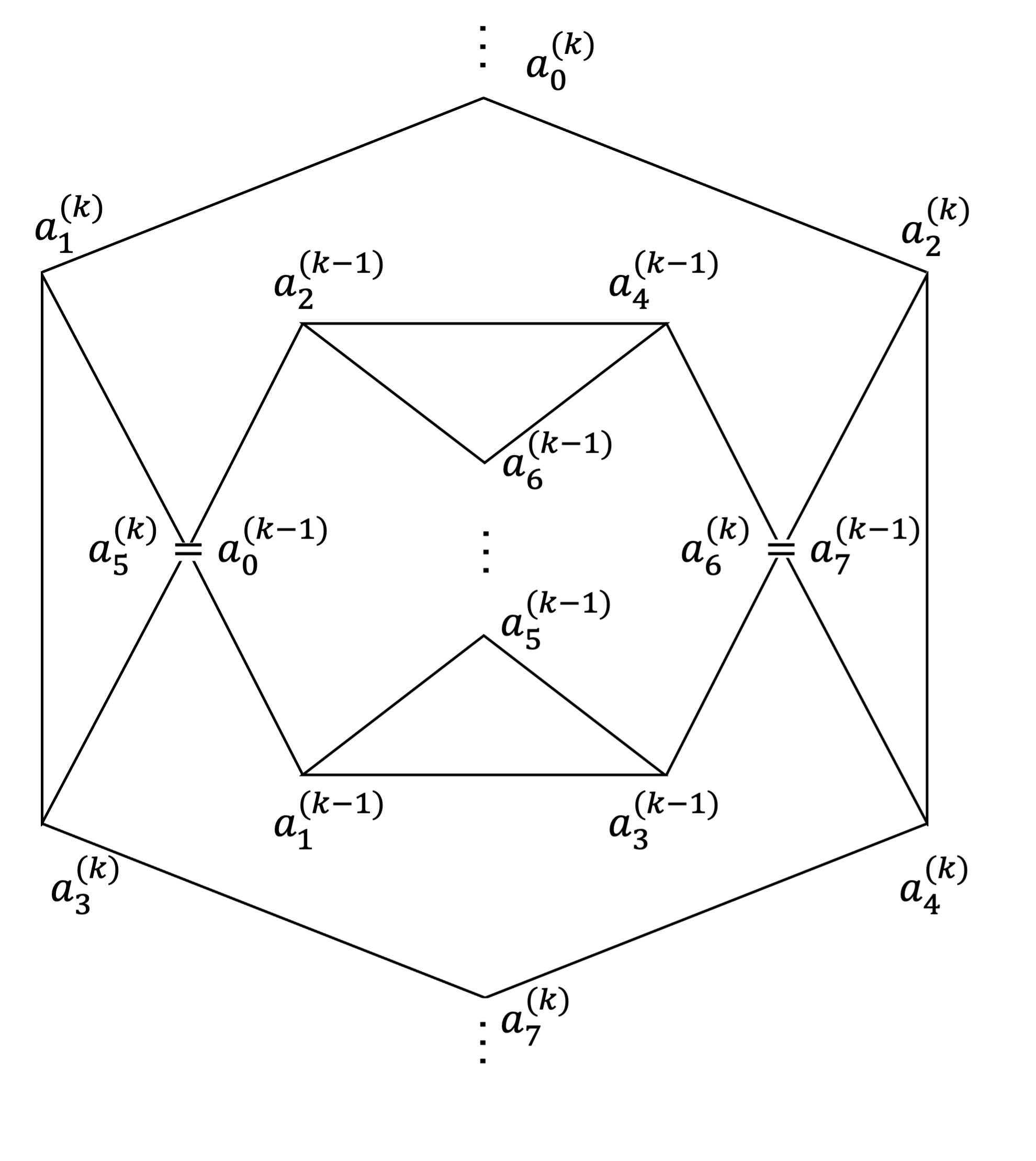}
  \caption{Nested Clifton graph, where [...] indicate further iterations.}
  \label{fig:3a}
\end{figure}

\begin{figure}[ht]
  \centering
  \includegraphics[width=.7\linewidth]{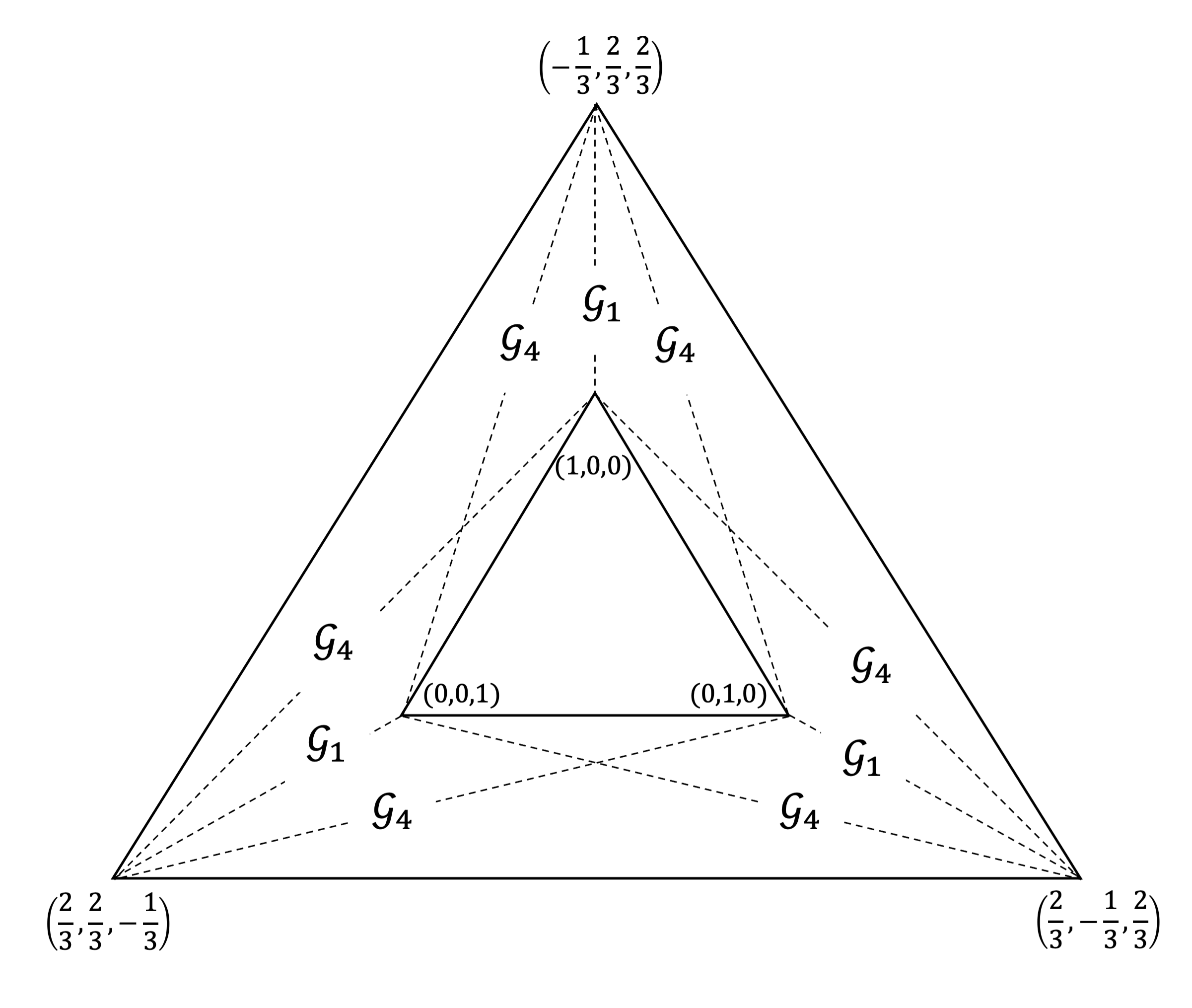}
  \caption{A 168-direction KS set. Solid lines join orthogonal directions. Dashed lines represent nested Clifton graphs.}
  \label{fig:3b}
\end{figure}

Using three ${\mathcal{G}}_{1}$ graphs and six ${\mathcal{G}}_{4}$ graphs, we construct the 168 direction KS set shown in \hyperref[fig:3b]{Figure 4}. The idea is simple: consider two orthogonal frames (triangles) and attach a nested Clifton graph between each direction (vertex) and the three directions of the other triangle. One of the three directions of a given frame must have outcome 1 by (KS.2), which in turn guarantees that the three directions of the other frame have outcome 0. This contradicts (KS.2) for the other frame, hence the set is~KS.

\section{Maximal Non-KS Sets}
Much of the literature on Kochen-Specker systems has focused on identifying minimal KS sets in three dimensions\footnote{The reason for which dimension three is special is two-fold. Every three-dimensional KS set can also be seen as a higher-dimensional one. Moreover, in dimensions four and higher, the size of the smallest KS set has already been found~\cite{pavivcic2005kochen}. In dimension three, the same question remains open, currently with a sizeable gap between Uijlen and Westerbaan's lower bound of 22 and Conway's graph of 31 directions.} \cite{arends2011searching,uijlen2016kochen,conway2006free}: Peres found a set with 33 directions \cite{A_Peres_1991}, and the smallest set to date consisting of 31 directions was found by Conway \cite{conway2006free}. By exhaustively constructing KS graphs and testing for topological embeddability, Uijlen and Westerbaan found a lower bound of 22 directions in three dimensions \cite{uijlen2016kochen}. The opposite question of finding maximal non-KS sets has not been explored. We do this here, starting with the definition of a non-KS set.

\subsection{Non-KS Sets}
\begin{definition}
A non-KS set in $d$-dimensions is a measurable subset $A \subseteq {S}^{d-1}$ that admits a valuation map $v:A \to \{0,1\}$ such that

\begin{enumerate}[(NKS.1)]
  \item $v(-n)=v(n)$ for all directions $n$ in $A$.
  \item $\sum\limits_{i \in I}v({n}_{i}) \leq 1$ for all sets of mutually orthogonal directions ${\{{n}_{i}\}}_{i \in I}$ in $A$.
  \item $\sum\limits_{i \in I} v({n}_{i}) = 1$ for all sets of $d$ mutually orthogonal directions ${\{{n}_{i}\}}_{i \in I}$ in $A$.
\end{enumerate}
\end{definition}

\textbf{Problem.} Find a non-KS set $A \subseteq {S}^{d-1}$ with maximal measure.

It is clear that any maximal non-KS set must be antipodally symmetric from the following theorem.

\begin{prop}
\label{Proposition 3.1.} The union of any non-KS set with its antipodal set is also a non-KS set.
\end{prop}
\begin{proof}
 Consider a non-KS set $A$ with some valuation map $v:A \to \{0,1\}$. Denote by $-A :=\{-n \in {S}^{d-1}:n \in A\}$ the antipodal set of $A$. Consider the map
\begin{align*}
  v':A \cup -A &\to \{0,1\} \\
  n &\mapsto \begin{cases}
                v(n) &n \in A \\
                v(-n) &n \in -A \backslash A
             \end{cases}
\end{align*}
This map trivially satisfies the three conditions in the definition of non-KS sets.
\end{proof}

In dimension $d=3$ the problem is equivalent to finding the largest subset such that no two orthogonal directions have value 1 and no three mutually orthogonal directions have value 0. It is instructive to first consider the simpler problems of identifying:
\begin{enumerate}[i)]
  \item the largest subset containing no two orthogonal directions;
  \item the largest subset containing no three mutually orthogonal directions.
\end{enumerate}

The first of these problems was posed by Witsenhausen in 1974 \cite{witsenhausen1974spherical}. The problem has not been solved but the following conjecture by Kalai and Wilson \cite{kalai2009large} has not yet been disproved.

\begin{conj}
The maximal measure for a subset of~$S^2$ containing no two orthogonal directions is attained by two open caps of geodesic radius $\pi/4$ around the south and north poles.
\end{conj}

The conjecture would imply that the measure is at most $1-1/\sqrt{2}$ as a fraction of the measure of the total sphere. DeCorte and Pikhurko \cite{decorte2016spherical} have proven that the maximal measure is at most $0.313$.

The second problem is an interesting variant of the first and is discussed by A. Perez \cite{perez2019generalizations}. In the next subsection we construct a large set with no three mutually orthogonal directions and conjecture that it is the largest such set.

\subsection{Constructing a Large Non-KS Set}

In order to construct a subset with no three mutually orthogonal directions, we begin with the union of two double caps. Since a double cap cannot contain two orthogonal directions, by the pigeonhole principle the union of two double caps cannot contain three mutually orthogonal directions (as one would have to contain two).

Without loss of generality, fix the poles of the double caps to be the positive and negative $x$ and $y$ axes. In terms of Cartesian coordinates, the boundaries of two of the caps on the upper (positive $z$) hemisphere have the following parametrizations:
\begin{equation}
    {\mathbf{r}}_{1}(t)=\left(\frac{1}{\sqrt{2}}, t,\sqrt{\frac{1}{2}-{t}^{2}}\right),
\end{equation}
\begin{equation}
    {\mathbf{r}}_{2}(t)=\left(u(t),\frac{1}{\sqrt{2}},\sqrt{\frac{1}{2}-{u(t)}^{2}}\right).
\end{equation}
Define the function $u(t)$ such that ${\mathbf{r}}_{1}(t)$ and ${\mathbf{r}}_{2}(t)$ are orthogonal:
\begin{equation}
    u(t)=-\frac{\frac{1}{2}+t}{1+t}~.
\end{equation}
Substituting into ${\mathbf{r}}_{2}(t)$ gives
\begin{equation}
    {\mathbf{r}}_{2}(t)=\left(-\frac{\frac{1}{2}+t}{1+t},\frac{1}{\sqrt{2}},\frac{1}{\sqrt{2}}\frac{\sqrt{\frac{1}{2}-{t}^{2}}}{1+t}\right).
\end{equation}
The third mutually orthogonal direction on the upper hemisphere is given by
\begin{equation}
    \begin{split}
        {\mathbf{r}}_{3}(t) &= {\mathbf{r}}_{1}(t) \times {\mathbf{r}}_{2}(t) \\
        &= \left(-\frac{1}{\sqrt{2}}\frac{\sqrt{\frac{1}{2}-{t}^{2}}}{1+t},-\sqrt{\frac{1}{2}-{t}^{2}},\frac{\frac{1}{2}+t+{t}^{2}}{1+t}\right)
    \end{split}
\end{equation}
The three parametrizations ${\mathbf{r}}_{1}(t),{\mathbf{r}}_{2}(t),{\mathbf{r}}_{3}(t)$  on the domain $t \in \left[-\frac{1}{2},0\right)$ together with their $\pi/2$ rotations about the $z$ axis define the boundary of a closed region $A \subseteq {S}^{2}$ which contains the north pole. We can parameterize the entire boundary through the function ${\mathbf{R}}(t)$ defined as follows. Let
\begin{equation}
    R_{\frac{\pi}{2}}:=
    \begin{pmatrix}
        0 &-1 &0\\
        1 &0 &0\\
        0 &0 &1
    \end{pmatrix}
\end{equation}
denote the rotation by $\pi/2$ about the $z$ axis. Using this, define
\begin{align}
  {\mathbf{r}}_0:\left[0,1/4\right) &\to S^2 \nonumber \\
  t &\mapsto \begin{cases}
                {\mathbf{r}}_1(6t-\frac{1}{2}) &t \in [\,0,\frac{1}{12}) \\
                R_{\frac{\pi}{2}}^3 {\mathbf{r}}_2(6t-1) &t \in [\, \frac{1}{12},\frac{1}{6}) \\
                R_{\frac{\pi}{2}}^2 {\mathbf{r}}_3(6t-\frac{3}{2}) &t \in [\, \frac{1}{6},\frac{1}{4})
             \end{cases}
\end{align}
and further,
\begin{align}
  {\mathbf{R}}:[\,0,1) &\to S^2 \nonumber \\
  t &\mapsto \begin{cases}
                {\mathbf{r}}_0(t) &t \in [\, 0,\frac{1}{4}) \\
                R_{\frac{\pi}{2}} {\mathbf{r}}_0(t-\frac{1}{4}) &t \in [\, \frac{1}{4},\frac{1}{2}) \\
                R_{\frac{\pi}{2}}^2 {\mathbf{r}}_0(t-\frac{1}{2}) &t \in [\, \frac{1}{2},\frac{3}{4}) \\
                R_{\frac{\pi}{2}}^3 {\mathbf{r}}_0(t-\frac{3}{4}) &t \in [\, \frac{3}{4},1)
             \end{cases}~.
\end{align}

It is convenient to extend the domain of $\mathbf{R}$ to the reals as $\forall t\in \mathbb{R}:\mathbf{R}(t)\equiv \mathbf{R}(t-\lfloor t \rfloor)$, making $\mathbf{R}$ a periodic function of period $1$.

\begin{lem}
\label{Lemma 3.2.} The following result holds:
\begin{equation}
    \forall t \in \mathbb{R}:\forall s \in \left[\, t+\frac{1}{3}, t+\frac{2}{3}\right):\textbf{R}(t)\cdot \textbf{R}(s)\leq 0
\end{equation}
\end{lem}
\begin{proof}
    The inequality follows from standard maximization of $\textbf{R}(t)\cdot \textbf{R}(s)$.
\end{proof}

\begin{lem}
\label{Lemma 3.3.} Consider the open octant $O(t)$ defined by the mutually orthogonal points $\mathbf{R}(t)$, $\mathbf{R}(t+\frac{1}{3})$ and $\mathbf{R}(t+\frac{2}{3})$ for some $t\in\mathbb{R}$:
\begin{equation}
    O(t) := \{\mathbf{r}\in S^2|\forall s\in\{0,\frac{1}{3},\frac{2}{3}\}:\mathbf{r}\cdot\mathbf{R}(t+s)>0\}.
\end{equation}
Then $\overline{O(t)}\subset A$.
\end{lem}
\begin{proof}
    We begin by showing $O(t) \cap \partial(A)=\emptyset$. For a contradiction, assume that $\exists u\in [\, 0,1):\mathbf{R}(u)\in O(t)$:
    \begin{equation}
    \begin{aligned}
        u\in [\, t+\frac{1}{3},t+\frac{2}{3}) \Rightarrow \mathbf{R}(u)\cdot\mathbf{R}(t)\leq 0 \Rightarrow \mathbf{R}(u)\notin O(t) \\
        u\in [\, t+\frac{2}{3},t+1) \Rightarrow \mathbf{R}(u)\cdot\mathbf{R}(t+\frac{1}{3})\leq 0 \Rightarrow \mathbf{R}(u)\notin O(t) \\
        u\in [\, t+1,t+\frac{4}{3}) \Rightarrow \mathbf{R}(u)\cdot\mathbf{R}(t+\frac{2}{3})\leq 0 \Rightarrow \mathbf{R}(u)\notin O(t)
    \end{aligned}
    \end{equation}

We have shown that $\forall u\in [\, t+\frac{1}{3},t+\frac{4}{3}):\mathbf{R}(u)\notin O(t)$, which is sufficient since $\mathbf{R}$ is periodic with period $1$. 
It follows trivially that:
\begin{equation}
    O(t)\subset\text{int}(A)\cup\text{ext}(A).
\end{equation}
Since $O(t)$ is connected and $(0,0,1)\in O(t)\cap\text{int}(A)$ it follows that
\begin{equation}
    O(t)\subset\text{int}(A)\subset A
\end{equation}
and finally
\begin{equation}
    \overline{O(t)}\subset\overline{A}=A.
\end{equation}
\end{proof}

\begin{lem}
\label{Lemma 3.4.} If two antipodal closed octants can perform complete rotations in the exterior of a region $U\subset S^2$, then $U$ contains no three mutually orthogonal directions.
\end{lem}
\begin{proof}
    Rather than the octants rotating, we take the dual perspective of fixing the octants and rotating the region~$U$.
    Define the fixed antipodal closed octants of~$S^2$, shown in \autoref{fig:5}, to be
    \begin{equation}
        \begin{split}
            \mathcal{O}_\pm := \{(\text{sin}(\theta)\text{cos}(\phi),\text{sin}(\theta)\text{sin}(\phi),\text{cos}(\theta)) \in S^2|~~~~~~~~\\
            (\theta,\phi)\in[\, 0,\pi/2]\times [\, 0,\pi/2]\ \cup\, [\, \pi/2,\pi]\times [\, \pi,3\pi/2]\, \}~.
        \end{split}
    \end{equation}
    For any region $U$ to perform a complete rotation in $\text{ext}(\mathcal{O}_\pm)$, each point in $U$ must cross the line
    \begin{equation}
        L:=\{(-\text{sin}(\gamma),\text{cos}(\gamma),0)\in S^2|\gamma\in(0,\pi/2)\}~.
    \end{equation}
    The set of points orthogonal to $p{=}(-\text{sin}(\gamma),\text{cos}(\gamma),0){\in}L$ on $S^2$ are
    \begin{equation}
    \begin{split}
        p^\perp:=\{(\text{sin}(\theta)\text{cos}(\gamma),\text{sin}(\theta)\text{sin}(\gamma),\text{cos}(\theta))\in S^2|& \\
        \phi\in [0,2\pi)\}&
    \end{split}
    \end{equation}
    and the set of such points contained in $\text{ext}(\mathcal{O}_\pm)$ are
    \begin{align}
    \begin{split}
        p^\perp\cap \text{ext}(\mathcal{O}_\pm)=\{(\text{sin}(\theta)\text{cos}(\gamma),\text{sin}(\theta)\text{sin}(\gamma),\text{cos}(\theta))\in S^2|& \\
        \theta\in (\frac{\pi}{2},\pi)\cup(\frac{3\pi}{2},2\pi)\}~,~~~~~~~~~&
    \end{split}
    \end{align}
    which contains no two orthogonal points. Hence, no point which crosses $L$ can be part of a mutually orthogonal set of three points in $U$.
\end{proof}

\begin{figure}[h!]
    \centering
    \includegraphics[width=.55\linewidth]{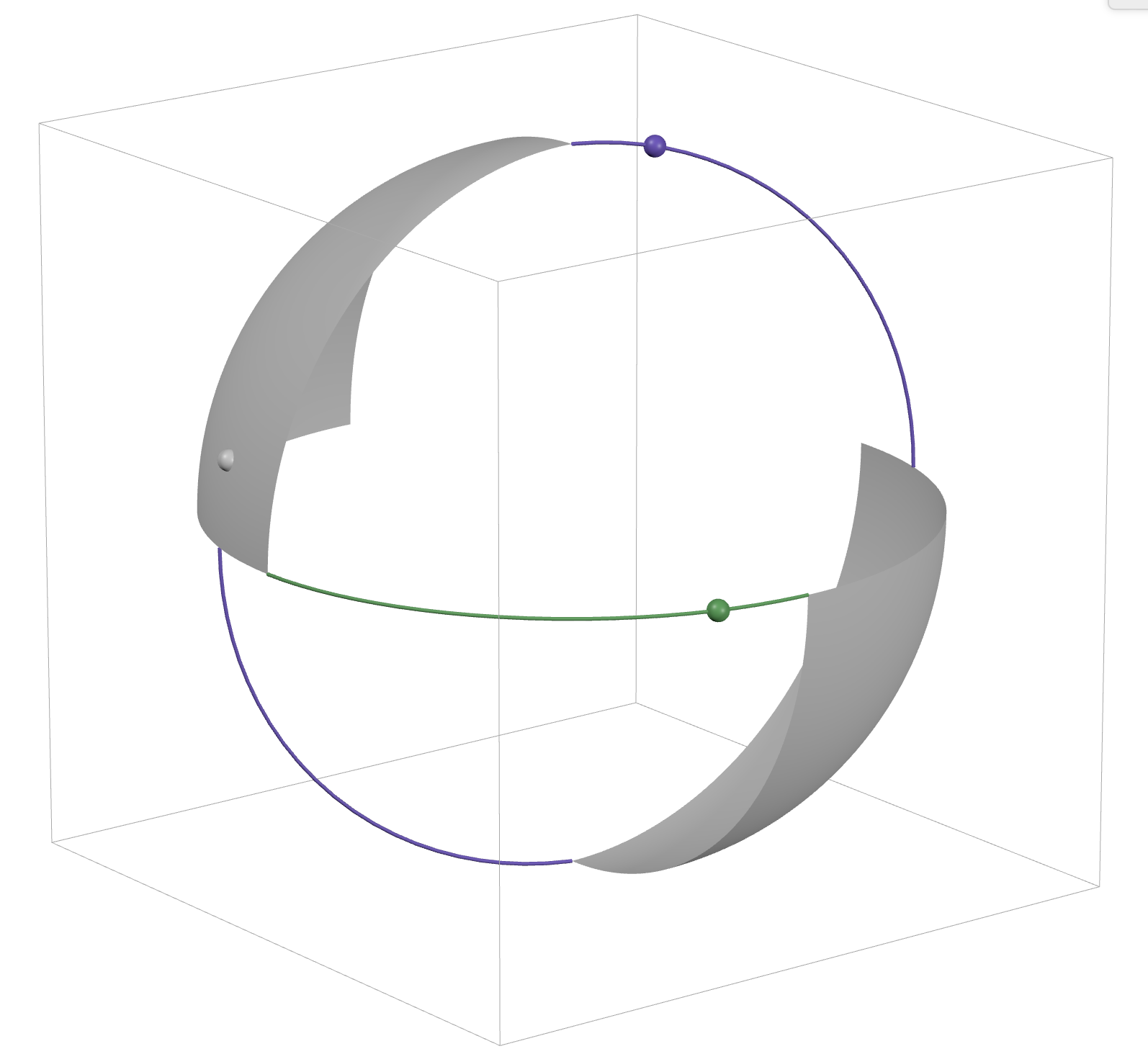}
    \caption{The grey shaded region depicts $\mathcal{O}_\pm$. The green line and point depict $L$ and $p$ respectively. The purple line depicts $p^\perp\cap \text{ext}(\mathcal{O}_\pm)$ and the purple and grey points depict two orthogonal points which cannot lie within $p^\perp\cap \text{ext}(\mathcal{O}_\pm)$.}
    \label{fig:5}
\end{figure}

\begin{prop}
\label{Proposition 3.5.} The region ${{N}_{0}}:={S}^2 \backslash (A \cup -A)$ contains no three mutually orthogonal directions.
\end{prop}
\begin{proof}
    It follows from Lemma~\ref{Lemma 3.3.} and the definition of $N_0$ that $(\overline{O(t)}\cup-\overline{O(t)})\cap N_0=\emptyset$. Hence, two antipodal closed octants can perform complete rotations in $\text{ext}(N_0)$. Hence, it follows from Lemma~\ref{Lemma 3.4.} that $N_0$ contains no three mutually orthogonal directions.
\end{proof}

\begin{remark}
{\upshape
A simpler proof for Proposition~\ref{Proposition 3.5.} is the following. Let $u, v, w\in N_0$ be a set of mutually orthogonal vectors. Rotate around $u$ until one or both $v$, $w$ land on the boundary of $N_0$. For concreteness, say $v$ lands on the boundary, while $w$ may be on the boundary or not after rotation. Then $u$ and $w$ must belong to the great circle $r\cdot v = 0$. However, by construction, the circle $r\cdot v = 0$ contains two quadrants which do not belong to $N_0$. Since $u$ is in $N_0$ and $w$ at most on the boundary, it follows that $u$ and $w$ cannot be perpendicular. }
\end{remark}

If we define ${{N}_{1}}$ to be the double cap with poles at the positive and negative $z$ axis then ${{N}_{1}} \cap {{N}_{0}} = \emptyset$ and the union ${N}:={{N}_{1}} \cup {{N}_{0}}$ forms a non-KS set with measure $|{N}|=|{{N}_{1}}| + |{{N}_{0}}|$ (see \autoref{fig:4}). The measures can be calculated by standard integration: 
\begin{equation}
\begin{aligned}
    |{{N}_{1}}|&=\frac{1}{4\pi}\left(2\int_{0}^{2\pi}\dd{\phi}\int_{0}^{\frac{\pi}{4}}\text{sin}(\theta)\ \dd{\theta}\right)  =1-\frac{1}{\sqrt{2}}~,\\
    |{{N}_{0}}|&=\frac{1}{2\pi}\left[8\int_{-\frac{1}{2}}^{0}\sqrt{\frac{1}{2}-{t}^{2}}\frac{\dd{}}{\dd{t}}\left(\text{arctan}(\sqrt{2}t))\right)\dd{t}\right.\\
    &\left.+4\int_{-\frac{1}{2}}^{0}\frac{\frac{1}{2}+t+{t}^{2}}{1+t}\frac{\dd{}}{\dd{t}}\left(\text{arctan}(\sqrt{2}(t+1))\right)\dd{t} \right] \\
    &=1-\frac{1}{\sqrt{2}}+\frac{\sqrt{2}}{\pi}\text{ln}(2)
\end{aligned}    
\end{equation}
where we used the relations $\phi=\text{arctan}(y/z)$, $\cos\theta = z$ for evaluating $|N_0|$. It follows that
\begin{equation}
    |{N}|=2-\sqrt{2}+\frac{\sqrt{2}\ln{2}}{\pi} \approx 0.8978,
\end{equation}
where the measures are normalized such that $|{S}^{2}|=1$. 

This gives a lower bound for the maximal non-KS set ${{N}_{\text{max}}}$:
\begin{equation}
    |{{N}_{\text{max}}}| \geq 2-\sqrt{2}+\frac{\sqrt{2}\ln{2}}{\pi} \approx 0.8978
\end{equation}

\begin{conj}
$N$ is the non-KS set of largest measure.
\end{conj}

\begin{figure}[!ht]
\centering
\begin{subfigure}{.5\textwidth}
  \centering
  \includegraphics[width=.55\linewidth]{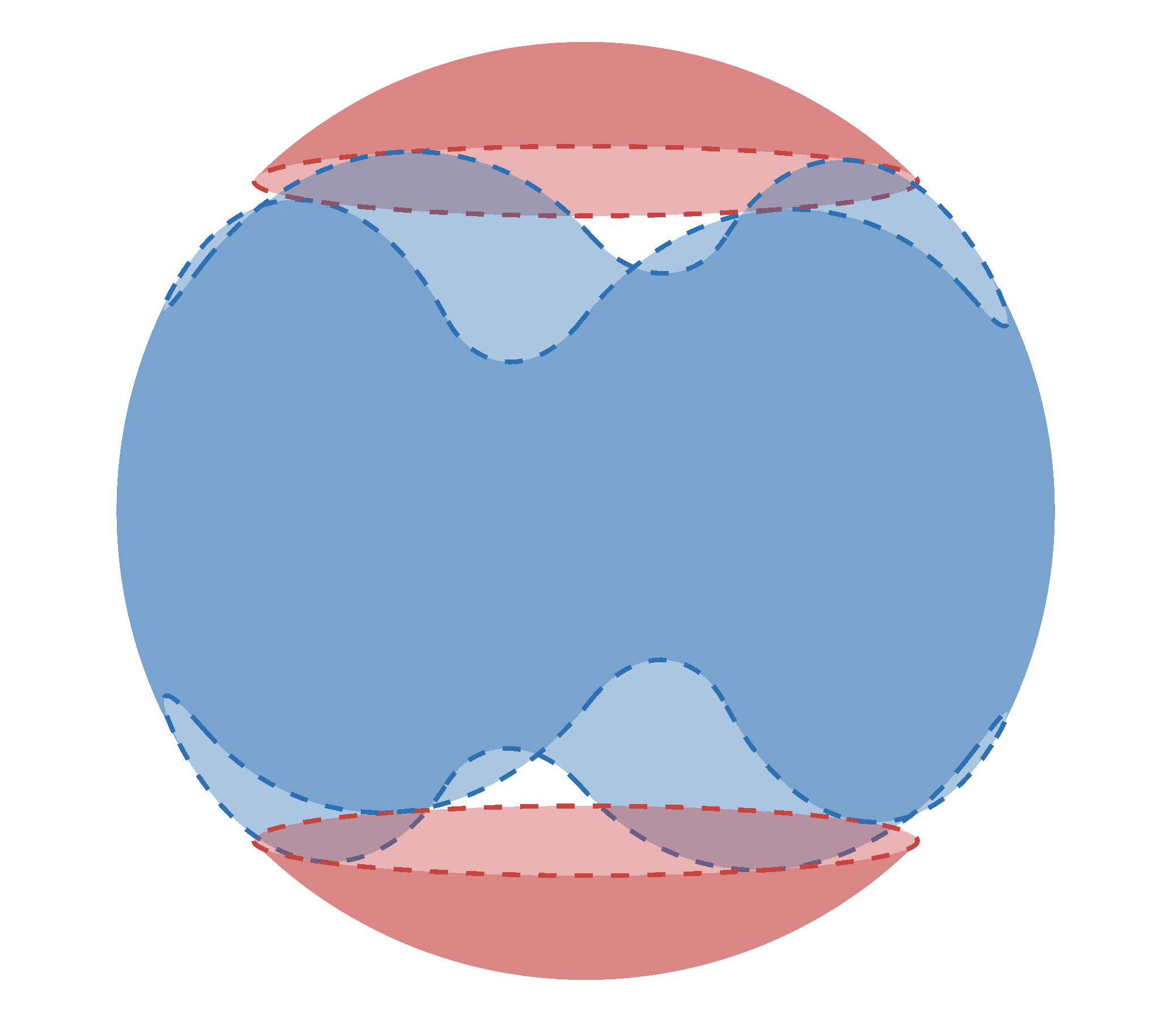}
  \caption{3D Plot}
  \label{fig:4a}
\end{subfigure}%
\\
\begin{subfigure}{.5\textwidth}
  \centering
  \includegraphics[width=.55\linewidth]{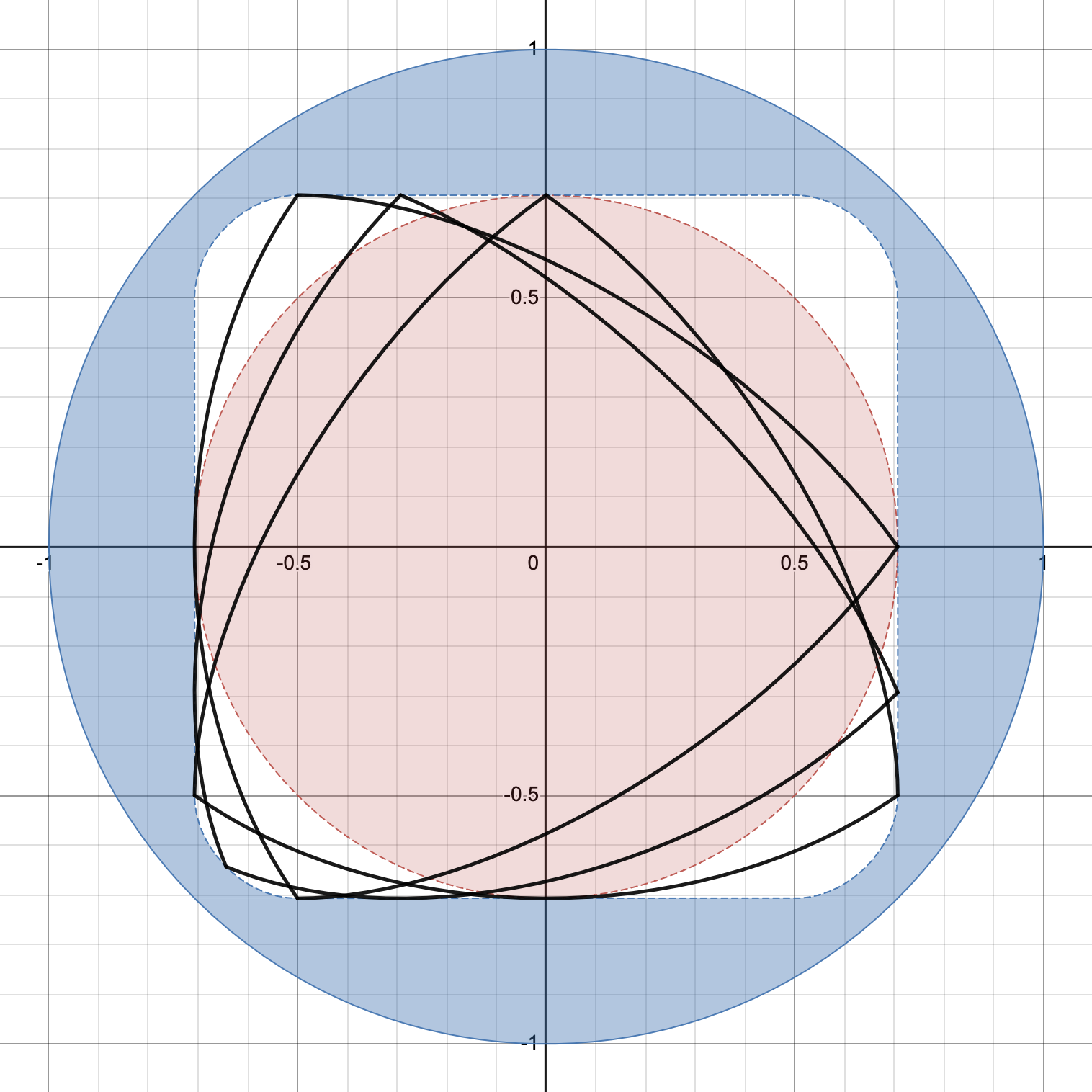}
  \caption{2D Projection onto $xy$-plane}
  \label{fig:4b}
\end{subfigure}
\caption{Plots of the non-KS set ${N}={{N}_{1}} \cup {{N}_{0}}$ where ${{N}_{1}}$ is plotted in red and ${{N}_{0}}$ in blue. The grey regions depict an octant rotating within $\text{ext}(N_0)$.}
\label{fig:4}
\end{figure}

\subsection{Maximal Non-KS Sets and Minimal KS Sets}

Say that a set $U \subseteq {S}^{d-1}$ avoids a set $P \subseteq {S}^{d-1}$ if there exists no orientation of $P$ that lies entirely inside $U$. By definition a non-KS set must avoid all KS sets. Hence a lower bound for the size of a KS set can be established by finding the smallest set of directions that any given non-KS set avoids. The following result by Alexandru Damian \cite{whitesides2015open} provides a lower bound for the minimal number of points that do not fit in an arbitrary subset $X \subseteq {S}^{d-1}$ of measure $|U|$.

\begin{prop}
\label{Theorem 3.7.} 
If a subset $U \subseteq {S}^{d-1}$ with measure $|U|$ avoids an $n$-point set $P$ then $n \geq \frac{1}{1-|U|}$.
\end{prop}
\begin{proof}
  Assume to the contrary that there is an $n$-point set $P$ with $n<\frac{1}{1-|U|}$ that $U$ avoids. Consider a random rotation of the set $P=\{{p}_{1},...,{p}_{n}\}$. Define a random indicator variable ${X}_{i}$ for each point ${p}_{i}$.
 \begin{equation}
  {X}_{i}  = \begin{cases}
                1 &{p}_{i} \in U \\
                0 &\text{otherwise}
             \end{cases}
 \end{equation}
 The expectation value for each ${X}_{i}$ is then
 \begin{equation}
     E\left[{X}_{i}\right] = P({X}_{i}=1) = |U| > \frac{n-1}{n}
 \end{equation}
 Defining by $X$ the total number of points that lie inside $U$ and computing its expectation using the linearity of expectation
 \begin{equation}
     E\left[X\right] = E\left[\sum\limits_{i=1}^{n}{X}_{i}\right] = \sum\limits_{i=1}^{n}E\left[{X}_{i}\right]=n|U|>n-1.
 \end{equation}
 Since $E\left[X\right]>n-1$ there must be at least one rotation such that all $n$ points lie inside $U$. So $U$ does not avoid $P$ which is a contradiction.
\end{proof}

\begin{theorem}
Consider a minimal KS set ${K}_{\text{min}}$ and a maximal non-KS set ${{N}_{\text{max}}}$. The following inequality must be satisfied
\begin{equation}
    |{K}_{\text{min}}| \geq \frac{1}{1-|{{N}_{\text{max}}}|}.
\end{equation}
In particular, in three dimensions it follows that
\begin{equation}
    |{K}_{\text{min}}| \geq \frac{1}{1-|{{N}_{\text{max}}}|} \geq \frac{1}{\sqrt{2}-1-\frac{\sqrt{2}\ln{2}}{\pi}} \approx 9.79.
\end{equation}
Hence a minimal KS set must contain at least 10 directions in three dimensions.
\end{theorem}
\begin{proof}
The result follows from the above discussion.
\end{proof}

It is interesting to observe the overlap of some small KS sets and our non-KS set $N$. Of course, the KS sets contain points that do not lie inside $N$ (in any orientation), but it is interesting to note, for example, that the Peres set is only just not covered by $N$ in the sense that $\overline{N}$ covers the Peres set (see \autoref{Peres_Conway_graphs}).

\begin{figure}[ht]
\centering
\begin{subfigure}{.5\textwidth}
  \centering
  \includegraphics[width=.5\linewidth]{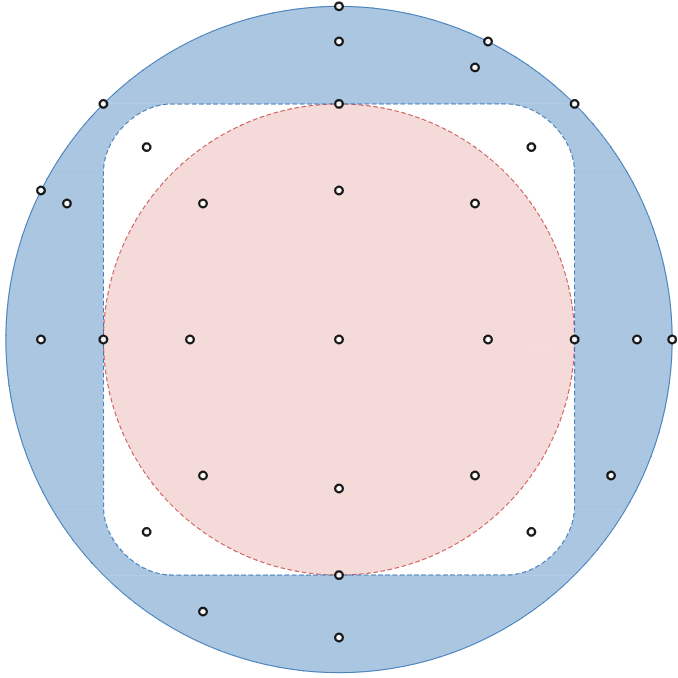}
  \caption{Conway's 31 direction KS set}
\end{subfigure}%
\\
\begin{subfigure}{.5\textwidth}
  \centering
  \includegraphics[width=.5\linewidth]{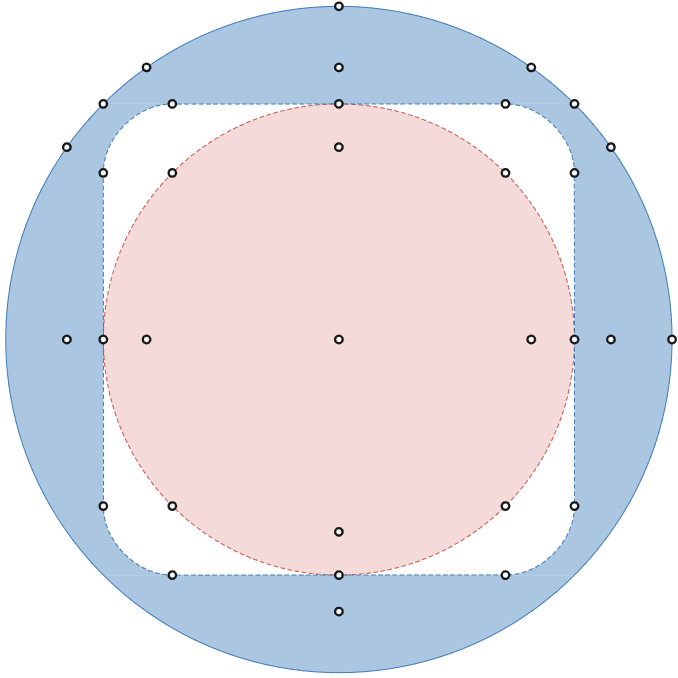}
  \caption{Peres' 33 direction KS set}
\end{subfigure}
\caption{Plots of Conway's 31 direction KS set and Peres' 33 direction KS set overlaid on the non-KS set $N$.}
\label{Peres_Conway_graphs}
\end{figure}


\begin{remark}
It is known~\cite{PhysRevLett.83.3751} 
that there exist rational non-KS sets which are dense in $S^2$. These sets form interesting examples of non-contextual configurations, however, their relevance for bounding the size of KS sets is limited due to their measure-zero size in $S^2$. Being of measure-zero implies that these rational non-KS sets occupy an infinitesimally small fraction of the sphere's surface and hence do not provide any meaningful bounds in the context of our inequalities. In the theorem above, the argument depends on measurable subsets $U \subset S^2$. Dense sets with zero measure, such as these rational configurations, do not impact the derived bound because their avoidance does not impose any constraints on $|K_{\text{min}}|$. 
\end{remark}

\section{Moving Sofa}

The original moving sofa problem, first formally discussed by Leo Moser \cite{moser1966moving} in 1966, is stated as follows:

\textbf{Problem.} What is the region of largest area which can be moved around a right-angled corner in a corridor of width one?

The area obtained is referred to as the sofa constant. The exact value of the sofa constant is an open problem, however there exist results on lower and upper bounds (see, e.g.~Refs.~\cite{gerver1992moving, kallus2018improved}).

Several variations of the moving sofa problem are said to have been considered by Conway, Shepard and other mathematicians \cite{stewart2004another} in the 1960's. One such variation is the ambidextrous moving sofa problem requiring the region to be able to make both right and left right-angled turns, which has since been considered by Romik~\cite{romik2018differential}.

There is an interesting connection between the set ${{N}_{0}}$ and a generalization of the moving sofa problem to a right-angled hallway defined on $S^2$ rather than $\mathbb{R}^2$.

\begin{prop}
A lower bound for the area of the largest shape that can be moved around a right-angled corner on the unit two-sphere is given by
\begin{equation}
|N_0|=1-\frac{1}{\sqrt{2}}+\frac{\sqrt{2}}{\pi}\text{ln}(2)
\end{equation}
\end{prop}
\begin{proof}
Consider $\text{ext}(\mathcal{O}_\pm)$ as defined in Lemma~\ref{Lemma 3.4.} to be the hallway, then it follows from Lemma~\ref{Lemma 3.3.} and the definition of $N_0$ that the region $N_0$ can be moved around $\text{ext}(\mathcal{O}_\pm)$. Hence, $|N_0|=1-\frac{1}{\sqrt{2}}+\frac{\sqrt{2}}{\pi}\text{ln}(2)$ is a lower bound for the sofa constant in this variation of the problem.
\end{proof}
\vspace{12pt}

\section{Conclusion}

In this paper we have constructed a conceptually simple KS set of 168 directions. We have also developed the notion of large measurable non-KS sets that can be KS coloured. We constructed a large non-KS set and calculated its measure to deduce a lower bound for the maximal measure of such a set. Relying on this lower bound, we used a probability argument to obtain a lower bound of 10 directions for the size of any KS set. Although this lower bound does not improve upon the current lower bound of 22 directions \cite{uijlen2016kochen}, it does not rely on exhaustively checking the existence of KS graphs. The probability argument of Proposition~\ref{Theorem 3.7.} is qualitatively different and does not take into account the geometry of the non-KS set. It is conceivable that including such information would lead to a stronger bound while avoiding the burden of exhaustive searches.
As a byproduct, we obtained a lower bound for the constant in the moving sofa problem on $S^2$.

{\bfseries Acknowledgements.} We would like to thank Carmen Constantin, Samson Abramsky, Curtis Bright, and Ad\'an Cabello for useful comments. We also thank the anonymous referee for valuable comments and suggestions, which have contributed to clarifying and improving the manuscript. TW would like to thank Oxford Physics for the award of the Gibbs Prize for the Best MPhys Research Project in 2023, on which this work is based. AC's research is supported by a Royal Society Dorothy Hodgkin Fellowship, grant number 231142.

\bibliography{bibliography}
\bibliographystyle{utcaps}

\end{document}